\documentclass[aps,pra,showpacs,twocolumn,superscriptaddress]{revtex4-1}
\usepackage{amsmath}
\usepackage{amsfonts}
\usepackage{amssymb}
\usepackage{graphicx}
\usepackage{dcolumn}
\usepackage{bm}
\usepackage{amsthm}
\usepackage[caption=false]{subfig}
\usepackage{natbib}

\begin{document}
\title{Fidelity of remote state preparation can be enhanced by local operation}

\author{Congyi Hua}
 \email{huacongyi@gmail.com}
 \affiliation{Zhejiang Institute of Modern Physics, Zhejiang University, Hangzhou 310027, China}
\author{Sheng-Wen Li}
 \email{lishengwen@yeah.net}
 \affiliation{Beijing Computational Science Research Center, Beijing 100084, China}
\author{Yi-Xin Chen}
 \email{Corresponding author\\ yxchen@zimp.zju.edu.cn}
 \affiliation{Zhejiang Institute of Modern Physics, Zhejiang University, Hangzhou 310027, China}

\begin{abstract}
Remote state preparation (RSP) is a quantum information protocol which allows preparing a quantum state at a distant location with the help of a preshared nonclassical resource state and a classical channel. The efficiency of successfully doing this task can be represented by the RSP-fidelity of the resource state. In this paper, we study the influence on the RSP-fidelity by applying certain local operations on the resource state. We prove that RSP-fidelity does not increase for any unital local operation. However, for nonunital local operation, such as local amplitude damping channel, we find that some resource states can be enhanced to increase the RSP-fidelity. We give the optimal parameter of symmetric local amplitude damping channel for enhancing Bell-diagonal resource states. In addition, we show RSP-fidelity can suddenly change or even vanish at instant under local decoherence.
\end{abstract}
\pacs{03.65.Yz, 03.67.-a, 03.65.Ta, 85.25.Cp}
\maketitle

\section{Introduction}
A quantum state of a qubit can be remotely prepared with the help of a preshared nonclassical bipartite resource state and a classical communication channel. This is called remote state preparation (RSP)~\cite{lo_classical-communication_2000,pati_minimum_2000}. The basic idea lies as follows~\cite{pati_minimum_2000}: Alice and Bob share a bipartite nonclassical state, for example, a maximally entangled qubit (ebit) in the ideal case. After a local measurement along certain direction at Alice's side, Alice sends Bob one classical bit (cbit), which tells Bob whether or not to flip his qubit, in order to obtain the state as desired. Several experiments for this scheme have been carried out, e.g., by nuclear magnetic resonance (NMR) techniques~\cite{peng_experimental_2003}.

In the ideal case, when the resource state is a maximally entangled state, the desired target state can be prepared perfectly at Bob's side, i.e., with fidelity equal to one. However, in practice, the resource state for RSP is usually a mixed one, thus we cannot always complete a perfect preparation. The efficiency of successfully doing this task can be measured by RSP-fidelity~\cite{dakic_quantum_2012}. Since RSP cannot be carried out without a nonclassical resource state, RSP-fidelity to some extent  describes the ``quantumness" of the bipartite resource state, which is similar with quantum discord~\cite{ollivier_quantum_2001}.

Recently, a surprising discovery has been made, which connects the RSP-fidelity with the geometric measure of quantum discord (GMQD)~\cite{dakic_quantum_2012}. They find that the necessary condition for a resource state to afford non-zero fidelity of RSP is that it must have nonzero quantum discord, rather than nonzero entanglement. Furthermore, they explicitly show that for a broad class of states, the RSP-fidelity is equal to the GMQD. Hence their work links GMQD to an operational meaning.

Some works point out that quantum discord can in-
crease via local operations~\cite{dakic_necessary_2010,streltsov_behavior_2011,hu_necessary_2012,piani_problem_2012}. For instance, a state with zero discord $\rho=(\vert 00\rangle \langle 00\vert+\vert11\rangle \langle 11\vert)/2$ can be transformed into a state with non-zero discord $\rho^\prime=(\vert 00\rangle \langle 00\vert+\vert+1\rangle \langle +1\vert)/2$ by a local operation acting on the first qubit such that $\Phi(\vert 0\rangle\langle 0\vert)=\vert 0\rangle\langle 0\vert$, $\Phi(\vert 1\rangle\langle 1\vert)=\vert +\rangle\langle +\vert$, with $\vert+\rangle =(|0\rangle +|1\rangle )/\sqrt{2}$. It can be verified that $\Phi$ is a trace-preserving and completely positive map. This fact makes us conjecture that RSP-fidelity, which is accessible in experiment, may share the similar property of local increase.

In this paper, we study the influence of local operations on the RSP-fidelity. We prove that the RSP-fidelity does not increase for unital local operation. However, for nonunital local operation, such as local amplitude damping channel, we find that some resource states can be enhanced to increase the RSP-fidelity. We give the optimal parameter of symmetric local amplitude damping channel for enhancing Bell-diagonal resource states. Meanwhile, we observe some interesting dynamic behaviors of RSP-fidelity. RSP-fidelity under local decoherence has the effect of sudden change and can vanish at instant. Particularly, the behavior of vanishing at instant indicates RSP-fidelity is more fragile than discord against decoherence, and therefore should take more enhancement.

We organize the remains of this paper as follows. In Sec.~\ref{secii}, we briefly review the RSP Protocol and the connection between RSP-fidelity and GMQD. In Sec.~\ref{seciii}, we study the influence of two types of local operations, the unital and the nonunital ones, on the RSP-fidelity. A necessary condition for local increase of RSP-fidelity is obtained. We also give a criterion of enhancibility for Bell-diagonal states under a symmetric local amplitude damping channel, as well as the optimal parameter of the channel. Finally, in Sec.~\ref{seciv}, we summarize our results.

\section{RSP}\label{secii}
In this section, we first introduce RSP based on the protocol in Ref.~\cite{pati_minimum_2000}. Then we make a quantitative comparison between RSP-fidelity and GMQD.
\subsection{RSP Protocal}
The concrete RSP protocol that our paper deals with comes from Pati~\cite{pati_minimum_2000}, where Alice intends to remotely prepare a pure qubit state from an ensemble represented by a great circle on the Bloch sphere using only one cbit and one preshared nonclassical bipartite resource state. The basic idea (also described by Daki{\'c} \textit{et al.}~\cite{dakic_quantum_2012}) is as follows: The target qubit, which is denoted as $\vert \phi\rangle$ or a Bloch vector $\boldsymbol{s}$, is chosen from the great circle of the Bloch ball orthogonal to a unit vector $\boldsymbol{b}$. A nonclassical state is preshared as resource state by both parties. Since in RSP Alice knows exactly what the target state is, she can choose such a unit vector $\boldsymbol{\alpha}$ to perform a local measurement along that would help to prepare the target state at Bob's side with the highest achievable fidelity~\cite{jozsa_fidelity_1994}. For instance, in the case when Alice and Bob have an singlet state $|\Psi ^-\rangle =\frac{1}{\sqrt{2}}\left(\vert\phi \rangle\vert\phi ^{\bot }\rangle -\vert\phi ^{\bot }\rangle \vert\phi \rangle \right)$ as resource state, Alice can perform a measurement along the direction $\boldsymbol{\alpha }=-\boldsymbol{s}$ and then send Bob the measurement outcome $\alpha =\pm 1$ as one cbit information. For $\alpha =-1$, Bob applies a $\boldsymbol{\pi}$ rotation about $\boldsymbol{\beta}$, and for $\alpha =1$, Bob does nothing. After this, the resulting state $\boldsymbol{r}$ on Bob's side will be nothing but the target state $\boldsymbol{s}$.

There is nothing special about sharing a singlet state as resource state. One can use any other maximally entangled states to achieve the same goal up to a change of measurement direction at Alice's side. In general, when the resource state is not maximally entangled, the resulting state $\boldsymbol{r}$ may differ from the target state $\boldsymbol{s}$. In order to evaluate the efficiency of RSP, Daki{\'c} \textit{et al.}~\cite{dakic_quantum_2012} introduce the notion of payoff-function $P\equiv (\boldsymbol{r}\cdot \boldsymbol{s})^2$. In each run, by deliberately choosing the measurement direction, Alice can optimize the payoff-function. We denote the optimized payoff-function by $P_{\text{opt}}$. The RSP-fidelity averages the optimized payoff-function and is minimized over all $\boldsymbol{\beta}$ on Bloch sphere, i.e., $F_{\text{RSP}}:=\underset{\boldsymbol{\beta }}{\inf }\left\langle P_{\text{opt}}\right\rangle.$ The RSP-fidelity captures two most natural aspects, $F_{\text{RSP}}=1$ ($P_{\text{opt}}\equiv 1$), for the resource state maximally entangled, while $F_{\text{RSP}}=0$ ($P_{\text{opt}}\equiv 0$) for the maximally mixed.

By representing a two-qubit resource state $\rho$ in terms of Pauli matrices $\{\sigma _1,\sigma _2,\sigma _3\}$
\begin{equation*}
\rho =\frac{1}{4}\left[I\otimes I+\boldsymbol{a}\cdot \boldsymbol{\sigma }\otimes I+I\otimes \boldsymbol{b}\cdot \boldsymbol{\sigma }+\sum _{k, l=1}^3 E_{k l}\sigma _k\otimes \sigma _l\right],
\label{eqbloch}
\end{equation*}
where $\boldsymbol{a}$, $\boldsymbol{b}$ are local Bloch vectors, the RSP-fidelity can be written as
\begin{equation}
F_{\text{RSP}}(\rho )=\frac{1}{2}\left(E_2^2+E_3^2\right).
\label{eqf}
\end{equation}
Here the coefficients $E_{k l}=\operatorname{tr}\left(\rho \sigma _k\otimes \sigma _l\right)$ form a real matrix denoted by $E$, and $E_1^2\geq E_2^2\geq E_3^2$ are the eigenvalues of $E^T E$. For simplicity, we will skip the derivation of this expression, one can consult Ref.~\cite{dakic_quantum_2012} for details. For RSP with only classical resource, where $E_2=E_3=0$, Non-zero RSP-fidelity is impossible. Since RSP-fidelity is a function of state, we can compare it with GMQD and this is what we are going to do in the next subsection.

\subsection{RSP-fidelity and GMQD}\label{seciib}
GMQD is introduced by the motivation for capturing total quantum correlation~\cite{dakic_necessary_2010}. The original paper defines GMQD as the square of the Hilbert-Schmidt distance from the given state to the set of classical-quantum states, i.e., $\underset{\chi }{\inf } \|\rho - \chi \|_{\text{HS}}^2,$
where $\|\rho - \chi \|_{\text{HS}}:=[\operatorname{tr} (\rho - \chi )^\dagger(\rho - \chi )]^{1/2}$, the infimum is taken over the set of classical-quantum states $\chi$. In order to make fair comparison with RSP-fidelity, we normalize the GMQD by a factor 2, as in Ref.~\cite{dakic_quantum_2012}, and deal with the normalized GMQD ($D_\text{G}$),
\begin{equation*}
D_\text{G}(\rho ):=2\, \underset{\chi }{\inf } \|\rho - \chi \|_{\text{HS}}^2.
\end{equation*}
Following Ref.~\cite{dakic_necessary_2010}, we have
\begin{align}
 D_\text{G}(\rho )&=\frac{1}{2}\left(|\boldsymbol{a}|^2+\|E\|_{\text{HS}}^2-\lambda _{\max }\right) \notag\\
&=\frac{1}{2}\left(\operatorname{tr} \boldsymbol{a} \boldsymbol{a}^{\text{T}}+\operatorname{tr} E E^{\text{T}}-\lambda _{\max }\right),
\label{eqdg}
\end{align}
with $\lambda _{\max}$ being the largest eigenvalue of $\boldsymbol{a} \boldsymbol{a}^{\text{T}}+E E^{\text{T}}$. In this paper, we also treat vectors as column matrices, like $\boldsymbol{a}$ in Eq.~(\ref{eqdg}).

Based on Eq.~(\ref{eqdg}), we can see for a broad set of states the (normalized) GMQD matches the RSP-fidelity. It is obvious that $\lambda _{\max }\leq \operatorname{tr} \boldsymbol{a} \boldsymbol{a}^{\text{T}}+E_1^2$, with equality if and only if $\boldsymbol{a}$ is parallel to the eigenvector corresponding to largest eigenvalue of $E E^{\text{T}}$. This set is a strict subset of the set of all X states~\cite{yu_evolution_2007} and is big enough to contain all the maximally mixed marginal states~\cite{horodecki_information-theoretic_1996}. Generally, 
\begin{equation*}
D_\text{G}(\rho)\geq F_{\text{RSP}}(\rho).
\end{equation*}
$D_\text{G}=0$ if and only if $\boldsymbol{a}=\boldsymbol{0}$ and $E_2^2=E_3^2=0$, so $D_\text{G}>0$ is a necessary but not sufficient condition for $F_{\text{RSP}}>0$. Fig.~\ref{fig1} shows a concrete example of $D_\text{G}>0$ but $F_{\text{RSP}}=0$. However, $F_{\text{RSP}}(\rho )$ and $D_\text{G}(\rho )$ simultaneously reaches $1$ when $\rho$ is maximally entangled. This is because, with respect to GMQD, the set of maximally discordant states is equivalent to that of maximally entangled states.

In addition, both $F_{\text{RSP}}(\rho)$ and $D_\text{G}(\rho)$ are invariant under local unitary transformation, i.e., for any unitary matrices $U_{1,2}$, $F_{\text{RSP}}\left(U_1\otimes U_2\rho U_1^\dagger\otimes U_2^\dagger\right)=F_{\text{RSP}}(\rho )$, $D_\text{G}\left(U_1\otimes U_2\rho U_1^\dagger\otimes U_2^\dagger\right)=D_\text{G}(\rho )$. Mathematically, performing a local unitary transformation is nothing but rechoosing a new local basis to represent the same density operator $\rho$, which obviously won't change $F_{\text{RSP}}(\rho )$ and $D_\text{G}(\rho )$.

\section{RSP under local operation}\label{seciii}
In this section, we study the behaviors of RSP-fidelity under local operation. For this purpose, we separately discuss the characteristics of the resource state under action of two types of quantum operations, the
unital and the nonunital ones. At last, we give a criterion for testing the enhancibility of a Bell-diagonal state under symmetric local amplitude damping channel and the optimal parameter of such channel.
\subsection{Unital operation}\label{seciiia}
Quantum operations or, from the viewpoint of decoherence, quantum channels are trace-preserving and completely positive maps, which are used for describing the dynamic changes to a state. A quantum operation $\Phi$ is said to be unital, if and only if $\Phi (I)=I$, or else it is nonunital. After introducing a representation of quantum operation, we will give the proposition that RSP-fidelity does not increase for any unital local operation.

Usually, a quantum operation is represented by Kraus operators. Here we deal with another representation for the convenience of the following proof~\cite{king_minimal_2001}. Recall that any qubit $\rho$ can be represented by a Bloch vector $\boldsymbol{r}$, so that $\rho =\frac{1}{2}\left(I+\boldsymbol{r}\cdot \boldsymbol{\sigma}\right)$, where $\boldsymbol{\sigma}$ is the vector of Pauli matrices. Accordingly, any single qubit operation $\Phi$, which maps a density matrix to another density matrix, can be represented by a unique $4\times 4$ matrix $\mathcal{T}$, which maps a Bloch vector to another one,
\begin{equation*}
\Phi:\frac{1}{2}\left(I,\boldsymbol{\sigma}^{\text{T}}\right)\left(\begin{array}{c}1 \\\boldsymbol{r} \\\end{array}\right)\mapsto\frac{1}{2}\left(I,\boldsymbol{\sigma}^{\text{T}}\right)\mathcal{T}\left(\begin{array}{c}1 \\\boldsymbol{r}\\\end{array}\right).
\end{equation*}
Here $\left(\begin{array}{c}1 \\\boldsymbol{r}\\\end{array}\right)$ is a $4\times 1$ matrix and $(I,\boldsymbol{\sigma}^{\text{T}})\equiv(I, \sigma_1, \sigma_2, \sigma_3)$, $\mathcal{T}$ has the form
\begin{align*}
\mathcal{T}=\left(\begin{array}{cc} 1 & \boldsymbol{0}^{\text{T}} \\ 
\boldsymbol{t} & T \\
\end{array}
\right),
\end{align*}
where $\boldsymbol{0}$ is zero vector, $\boldsymbol{t}$ and $3\times 3$ matrix $T$ are real. For a unital operation $\Phi$, we have $\boldsymbol{t}=\boldsymbol{0}$. Using the singular value decomposition, we can write
\begin{equation*}
T=O_1D O_2^{\text{T}}=R_1(\pm D) R_2^{\text{T}},
\end{equation*}
where $O_{1,2}$ are orthogonal matrices, $D=\text{diag}\{D_{11}, D_{22}, D_{33}\}$ and $R_{1,2}$ are rotations. A common convention is to list $\{D_{11}, D_{22}, D_{33}\}$ in descending order. In this case, the diagonal matrix $D$ is unique. Define the map $\Phi _D$ by
\begin{equation}
 \Phi _D\left(\rho\right):=\frac{1}{2}\left(I,\boldsymbol{\sigma}^{\text{T}}\right)\left(
\begin{array}{cc}
 1 & \boldsymbol{0}^{\text{T}} \\
 \boldsymbol{d} & \pm D \\
\end{array}
\right)\left(\begin{array}{c}1 \\\boldsymbol{r}\\\end{array}\right),
\label{eqphid}
\end{equation}
with $\boldsymbol{d}=\left(d_1,d_2,d_3\right)^{\text{T}}=R_2 R_1^{\text{T}}\boldsymbol{t}$.
Since every rotation on a Bloch vector is equivalent to a unitary transformation on the respective density operator, the general quantum operation $\Phi$ can be factorized into simpler parts as
\begin{equation}
\Phi (\rho )=U \Phi _D(V \rho  V^\dagger)U^\dagger,
\label{eqdecomp}
\end{equation}
where $U$ and $V$ are unitary matrices corresponding to $R_1$ and $R_2$.

Now we come to the following proposition.
\newtheorem*{prop}{Proposition}
\begin{prop}
RSP-fidelity of two-qubit states do not increase under local unital operations.
\end{prop}
\begin{proof}
According to Eq.~(\ref{eqdecomp}), an arbitrary local unital operation $\Phi$ on a two-qubit state $\rho$ can be given as
\begin{equation*}
\Phi(\rho )=\left(U_1\otimes U_2\right)\Phi _D\left(V_1\otimes V_2\rho  V_1^\dagger\otimes V_2^\dagger\right)\left(U_1^\dagger\otimes U_2^\dagger\right),
\end{equation*}
where $U_{1, 2}$, $V_{1, 2}$ are unitary and $\Phi _D$ is represented by the matrix
\begin{align*}
\notag
 \left(
\begin{array}{cc}
 1 & \boldsymbol{0}^{\text{T}} \\
 \boldsymbol{0} & \pm D_1 \\
\end{array}
\right)\otimes\left(
\begin{array}{cc}
 1 & \boldsymbol{0}^{\text{T}} \\
 \boldsymbol{0} & \pm D_2 \\
\end{array}
\right).
\end{align*}
Here we have used the fact that, for a unital operation, $\boldsymbol{d}=\boldsymbol{0}$. The trace-preserving and completely positive property of $\Phi$ requires both the points represented by $\left(\left(D_1\right)_{11},\left(D_1\right)_{22},\left(D_1\right)_{33}\right)$, $\left(\left(D_2\right)_{11},\left(D_2\right)_{22},\left(D_2\right)_{33}\right)$ are constrained inside a tetrahedron with corners $(1,1,1)$, $(1,-1,-1)$, $(-1,1,-1)$ and $(-1,-1,1)$~\cite{king_minimal_2001}.
Since $F_{\text{RSP}}$ is invariant under local unitary transformation (see Sec.~\ref{seciib}), we can restrict our considerations to those local unital operations which have $\Phi _D$-form as in Eq.~(\ref{eqphid}). For any two-qubit state $\rho$, a simple calculation shows
\begin{align*}
\Phi _D(\rho)=&\frac{1}{4}\left\{\vphantom{\sum _{i,j=1}^3}I\otimes I+\left(D_1 \boldsymbol{a}\right)\cdot \boldsymbol{\sigma }\otimes I+I\otimes\left(D_2\boldsymbol{b}\right)\cdot \boldsymbol{\sigma }\notag\right.\\
&\left.+\sum _{i,j=1}^3 \left(D_1E D_2\right)_{i j}\sigma _i\otimes \sigma _j\right\},
\end{align*}
with $\left(D_1E D_2\right)_{i j}\equiv\left(D_1\right)_{i i}\left(D_2\right)_{j j}E_{i j}.$

From Eq.~(\ref{eqf}), we have
\begin{align}
 F_{\text{RSP}}(\rho )={}&\frac{1}{2}\left[\operatorname{tr}(E^{\text{T}}E)-\underset{R}{\sup }(\boldsymbol{k}^{\text{T}}R^{\text{T}}E^{\text{T}}E R \boldsymbol{k})\right]\notag \\
\begin{split}
={}&\frac{1}{2}\underset{R}{\inf}\left\{\operatorname{tr}\left[\left(R^{\text{T}}E R\right)^{\text{T}}R^{\text{T}}E R\right]\right.\\
&\left.-\boldsymbol{k}^{\text{T}}\left(R^{\text{T}}E R\right)^{\text{T}}R^{\text{T}}E R \boldsymbol{k}\right\},
\end{split}
\label{eqproof}
\end{align}
and similarly, have
\begin{align*}
F_{\text{RSP}}\left(\Phi _D(\rho )\right)
={}&\frac{1}{2}\underset{R}{\inf }\left\{\operatorname{tr}\left[\left(R^{\text{T}}D_1E D_2 R\right)^{\text{T}}R^{\text{T}}D_1E D_2 R\right]\right.\\
&\left.-\boldsymbol{k}^{\text{T}}\left(R^{\text{T}}D_1E D_2 R\right)^{\text{T}}R^{\text{T}}D_1E D_2 R \boldsymbol{k}\right\}
\end{align*}
where $\boldsymbol{k}=(1,0,0)^{\text{T}}$, the supremum and infimum are taken over the set of all rotation matrices.
We assume the infimum in Eq.~(\ref{eqproof}) is achieved for $R=R_0$ and define $E_0\text{:=}R_0^{\text{T}}E R_0$. Then we can see
\begin{align*}
&F_{\text{RSP}}\left(\Phi _D(\rho )\right)\\
\leq{}& \frac{1}{2}\left\{\operatorname{tr}\left[\left(R_0^{\text{T}}D_1E D_2 R_0\right)^{\text{T}}R_0^{\text{T}}D_1E D_2 R_0\right]\right.\\
&\left.-\boldsymbol{k}^{\text{T}}\left(R_0^{\text{T}}D_1E D_2 R_0\right)^{\text{T}}R_0^{\text{T}}D_1E D_2 R_0 \boldsymbol{k}\right\} \\
={}&\sum _{i=1}^3 \sum _{j=1}^2 \frac{1}{2}\left(D_1\right)_{i i}^2\left(D_2\right)_{j j}^2\left(E_0\right)_{i j}^2 \\
\leq{}&F_{\text{RSP}}(\rho )
\end{align*}
We find that, for a unital operation $\Phi$, $F(\Phi (\rho ))=F(\rho)$ if and only if $\Phi_D$ is trivial identity.
\end{proof}
The above proposition excludes the possibility of enhancing the RSP resource with unital local operations, such as the direct product of single qubit depolarizing, bit flip, phase flip and bit-phase flip channels.

\subsection{Nonunital operation}\label{seciiib}
In this subsection, we transfer our attention to symmetric local amplitude damping channel to see if a suitable local nonunital operation can bring some enhancement to the resource states.

Amplitude damping channel is a typical kind of nonunital operation that characterizes the behavior of energy dissipation in a lot of concrete physical processes, such as spontaneous emission of an atom, behavior of a spin system approaching equilibrium with the environment at high-temperature. Although, in quantum information theory, the concept of amplitude damping channel originates in quantum noise processes, amplitude damping channels are experimentally achievable~\cite{salles_experimental_2008,qing_linear_2007,almeida_environment-induced_2007}.

Given an initial qubit $\rho$, the action of a one-qubit quantum operation on it can be given by $\Phi(\rho )=\sum _{i} E_i\rho  E_i^\dagger$, where $E_i$ are the Kraus operators. The one-qubit amplitude damping channel has the Kraus form $E_0=\left(\begin{array}{cc}1 & 0 \\ 0 & \sqrt{q} \\\end{array}\right)$, $E_1=\left(\begin{array}{cc}1 & \sqrt{p} \\ 0 & 0 \\\end{array}\right)$, with $q=1-p$, $q\in[0, 1]$. The parameter $p$ is responsible for a wide range of physical phenomena. There exist some time-dependent decoherence models, where $p$ is replaced by a time-varying function $1-e^{-\Gamma  t}$ with $\Gamma$ a constant characterizing the speed of the processes.

For simplicity, we only consider the symmetric situation in which the damping rates on both sides are equal. In addition, we restrict the resource states for RSP to Bell-diagonal states, which are diagonal in the Bell basis ($\vert 00\rangle\pm\vert 11\rangle$, $\vert 01\rangle\pm\vert 10\rangle$). A Bell-diagonal state $\rho$ can also be expressed by Pauli matrices as
\begin{align}\label{eqbelldiag}
\rho =\frac{1}{4}\left(I+\sum _{i=1}^3 c_i\sigma _i\otimes \sigma _i\right).
\end{align}
The condition for $\rho$ being a proper density matrix imposes that the point made of coordinates $\left(c_1, c_2, c_3\right)$ must lie inside the tetrahedron with the four Bell states being the corners (see Fig.~\ref{fig2}(a))~\cite{horodecki_information-theoretic_1996}. The Bell-diagonal states may arise in a wide variety of physical situations and play an important role in quantum information processings~\cite{verstraete_optimal_2003,bennett_purification_1996,wang_experimental_2006,koashi_quantum_1996}. Moreover, there always exists a local transformation that can transform a given mixed state to the corresponding Bell-diagonal form~\cite{cen_local_2002}.

Given an initial two-qubit state $\rho$, its evolution under a amplitude damping channel can be modeled in the Kraus form
\begin{align*}
\Phi _p (\rho )=\sum _{i,j=0}^1 E_i\otimes E_j\rho  E_i^\dagger\otimes E_j^\dagger.
\end{align*}
We can calculate the output of an arbitrary Bell-diagonal state with the form of Eq.~(\ref{eqbelldiag}) under a symmetric amplitude damping channel with parameter $p$, and arrive at
\begin{align}\label{equnderch}
\begin{split}
\Phi _p (\rho )={}&\frac{1}{4}\left[I+p \left(I\otimes \sigma _3+\sigma _3\otimes I\right)+ q c_1\sigma _1\otimes \sigma _1\right.\\
&\left.+q c_2\sigma _2\otimes \sigma _2+ \left(c_3 q^2+p^2\right)\sigma _3\otimes \sigma _3\right].
\end{split}
\end{align}
Calculate both the RSP-fidelity and the GMQD for $\Phi _p (\rho )$, we have
\begin{align}
\begin{split}
F_{\text{RSP}}\left(\Phi _p (\rho )\right)=\frac{1}{2}\left[q^2\left(c_1^2+c_2^2\right)+\left(c_3 q^2+p^2\right)^2\right.\\
\left.-\max \left\{\left(q c_1\right)^2,\left(q c_2\right)^2,\left( c_3 q^2+p^2 \right)^2 \right\}\right],
\end{split}
\label{eqfunderch}
\end{align}
\begin{align}
\begin{split}
D_\text{G}&\left(\Phi _p (\rho )\right)=\frac{1}{2}\left[q^2\left(c_1^2+c_2^2\right)+\left(p^2 + c_3 q^2\right)^2+ p^2\right.\\
&\left.-\max \left\{\left(q c_1\right)^2,\left(q c_2\right)^2,\left(p^2 + c_3 q^2\right)^2 + p^2\right\}\right].
\end{split}
\label{eqdgunderch}
\end{align}

\begin{figure}
	\includegraphics[width=0.47\textwidth]{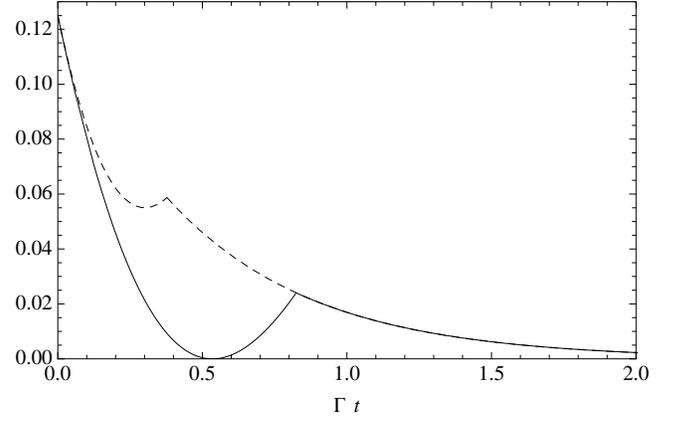}
	\caption{Dynamics of GMQD (dashed line) and RSP-fidelity (solid line) as a function of $\Gamma \text{\textit{$ $}}t$. The initial state used here is given by Eq.~(\ref{eqbelldiag}) with $c_1=0.5$, $c_2=0$ and $c_3=-0.5$.}
	\label{fig1}
\end{figure}

In Fig.~\ref{fig1}, we plot both the RSP-fidelity and the GMQD of $\Phi _p (\rho )$ as functions of scaled time $\Gamma  t$. The initial resource state used is given by Eq.~(\ref{eqbelldiag}) with $c_1=0.5$,  $c_2=0$ and $c_3=-0.5$. In this case, although the resource state is initially Bell-diagonal, but it cannot maintain the Bell-diagonal form as long as it's evolution begins. This is why the curves of $D_\text{G}$ and $F_{\text{RSP}}$ do not always match. We can observe that $D_\text{G}\geq F_{\text{RSP}}$ is always satisfied, which coincides with the result in Sec.~\ref{seciib}. And just as we expected, we find both RSP-fidelity and GMQD can increase during the time evolution under certain conditions. Also, due to the maximization procedures in Eqs.~(\ref{eqfunderch}) and~(\ref{eqdgunderch}), both RSP-fidelity and GMQD have a \textit{sudden change} effect during the time evolution. The phenomena of sudden change in the dynamics of both GMQD and quantum discord have been discussed in detail in Ref.~\cite{maziero_classical_2009,lu_geometric_2010,yao_geometric_2012}. Yet most interestingly, RSP-fidelity vanishes at instant. This is quite different from discord and may show one different facet of RSP-fidelity from GMQD. As we know, when two qubits subject to independent Markovian decoherence, the evolution of discord  decays only in asymptotic time~\cite{modi_classical-quantum_2012,werlang_robustness_2009}. In this sense, compared with quantum algorithms based only on quantum discord, the enhancement of RSP-fidelity is more worthy of consideration. But it should be pointed out, in the situation of non-Markovian, quantum discord can vanishes, see~\cite{fanchini_non-markovian_2010,alipour_quantum_2012}.

\subsection{The enhancibility and the optimal parameter}
\label{seciiic}
In the above subsection, we have considered the symmetric local amplitude damping channel as decoherence noise, which lands us in a passive position. However since the amplitude damping channel is the quantum operation which we can experimentally control, we can impose it on the resource state for enhancing as our aim. In this subsection, we will give the criterion for testing the enhancibility of a Bell-diagonal state under symmetric local amplitude damping channel, as well as the optimal parameter of such channel. 

For this purpose, we need to thoroughly discuss the expression of $F_{\text{RSP}}$. Since the maximization procedure in Eq.~(\ref{eqfunderch}) complicates the calculation, we divide the discussion into two cases, depending on the sign of $c_3$.

\begin{figure*}
	\centering
	\subfloat[]{\includegraphics[width=0.4\textwidth]{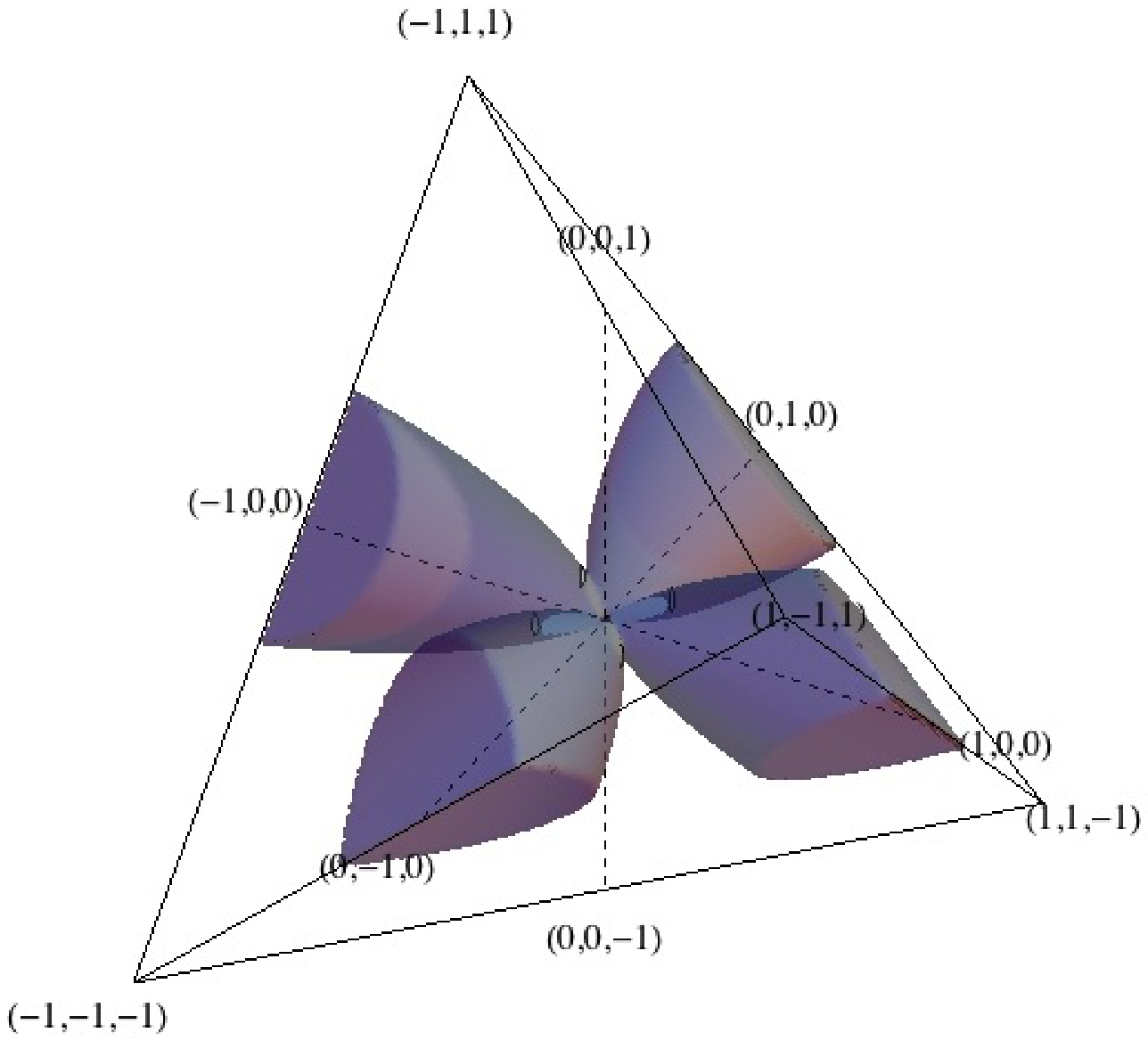}}
	\subfloat[]{\includegraphics[width=0.4\textwidth]{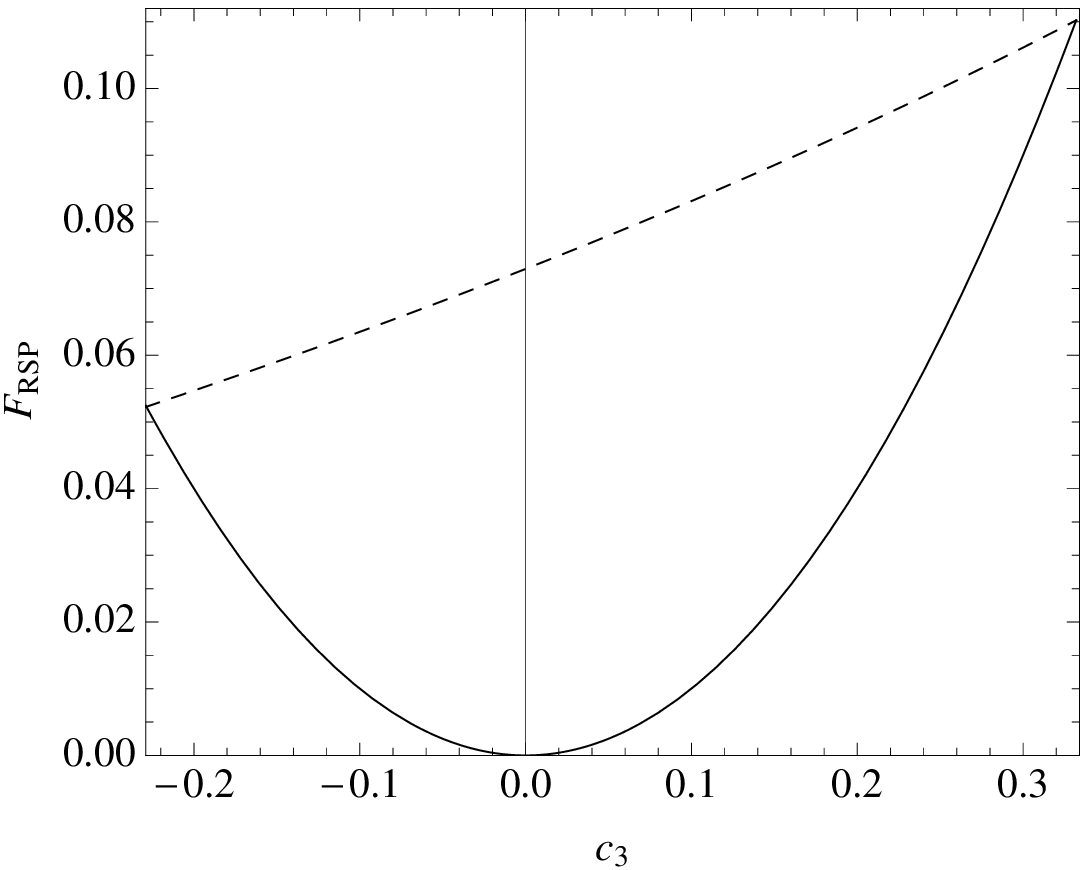}}	
	\caption{(a) The tetrahedron represents the set of Bell-diagonal states where the four Bell states are the corners. The zero-discord states are labeled by the dashed lines. The dark region consisting four part with the same shape represents the states which satisfy condition (11). (b) The RSP-fidelity of RSP resource states (initially Bell-diagonal) before (solid line), and after (dashed line) applying the optimal symmetric local amplitude damping channel.  We have set, in this figure, $c_2=-1$, $c_3=c_1$.}
	\label{fig2}
\end{figure*}

\subsubsection{Case: $c_3\geq 0$}
We set $c=\max \left\{\left|c_1\right|, \left|c_2\right|\right\}$. To get rid of the maximum operator in Eq.~(\ref{eqfunderch}), we need to note the sign of $\left|p^2+c_3q^2\right|-|q c|$. When $c_3\leq c$, by observing the spatial relation between the straight line segment $y=q c$, $q\in [0,1]$ and the curve segment $y=p^2+c_3q^2$, $q\in [0,1]$, we can easily express $F_{\text{RSP}}\left(\Phi _p (\rho )\right)$ from Eq.~(\ref{eqfunderch}) in piecewise form,
\begin{equation}
\begin{split}&F_{\text{RSP}}\left(\Phi _p (\rho )\right)\\={}&\begin{cases}
 \frac{1}{2}\left[q ^2\left(c_1^2+c_2^2-c^2\right)+\left(c_3 q^2+p^2\right)^2\right] & q_1\leq q\leq 1, \\
 \frac{1}{2}q ^2\left(c_1^2+c_2^2\right) & 0\leq q<q_1.
\end{cases}\end{split}
\label{eqfpiecewise}
\end{equation}
Here
\begin{equation}
q_1=\frac{2}{2+c+\sqrt{c^2+4\left(c-c_3\right)}}
\end{equation}
 is the the smaller root of the equation $q c=c_3 q^2+(1-q)^2$. Subsequently, we treat $q$, instead of $p$, as the independent variable of $F_{\text{RSP}}\left(\Phi _q (\rho )\right)$. For $q\in \left[q_1,1\right]$,
\begin{equation}
\begin{split}\frac{\text{d}}{\text{d} q}F_{\text{RSP}}&\left(\Phi _q (\rho )\right)=\left(c_1^2+c_2^2-c^2\right)q\\&+\left[(1-q)^2+c_3q^2\right]\left[2\left(c_3+1\right)q-2\right].
\end{split}
\label{eqdf}
\end{equation}
We can see from Eq.~(\ref{eqfpiecewise}) that local maximums of $F_{\text{RSP}}\left(\Phi _q (\rho )\right)$ will not occur for $q\in \left[0,q_1\right)\cup \left[\frac{1}{1+c_3},1\right]$. While for $q\in \left[q_1,\frac{1}{1+c_3}\right)$, through an analysis of Eq.~(\ref{eqdf}), we observe that $q=q_1$ is the only possible local maximum point of $F_{\text{RSP}}\left(\Phi _q (\rho )\right)$. By inserting $q=q_1$ into Eq.~(\ref{eqfpiecewise}), the existence of a symmetric local amplitude damping channel of parametrer $q_1$ for which $F_{\text{RSP}}\left(\Phi _{q_1} (\rho )\right)>F_{\text{RSP}}(\rho )$ can be checked by
\begin{equation}\frac{c_2^2+c_1^2}{c_1^2+c_2^2+c_3^2-c^2}>\frac{\left(2+c+\sqrt{c^2+4\left(c-c_3\right)}\right)^2}{4}.
\label{eqtest}
\end{equation}

In the above discussion, we considered the situation of $c_3\leq c$. When $c<c_3$, $F_{\text{RSP}}\left(\Phi _q (\rho )\right)<F_{\text{RSP}}(\rho )$ holds for all value of $q\in [0,1)$.

\subsubsection{Case: $c_3<0$}
When $-c_3\leq c$, the straight line segment $y=q c$, $q\in [0,1]$ and the curve segment $y=\left|p^2+c_3q^2\right|$, $q\in [0,1]$ intersect only once at $q=q_1$. Thus by Eq.~(\ref{eqfunderch}), the expressions of $F_{\text{RSP}}\left(\Phi _q (\rho )\right)$ and $\frac{\text{d}}{\text{d} q}F_{\text{RSP}}\left(\Phi _q (\rho )\right)$ in this case are the same with Case 1. And for $q\in \left[q_1,1\right]$, a short analysis of Eq.~(\ref{eqdf}) shows that $q=q_1$ is the only possible local maximum point of $F_{\text{RSP}}\left(\Phi _q (\rho )\right)$. Hence we arrived at the same fomula as Eq.~(\ref{eqtest}) in Case 1 for checking the existence of a symmetric local amplitude damping channel for enhancing $\rho$,
\begin{equation*}
\frac{c_2^2+c_1^2}{c_1^2+c_2^2+c_3^2-c^2}>\frac{\left(2+c+\sqrt{c^2+4\left(c-c_3\right)}\right)^2}{4}.
\end{equation*}

Then, when $c<-c_3$, $F_{\text{RSP}}\left(\Phi _q (\rho )\right)<F_{\text{RSP}}(\rho )$ holds for all value of $q\in [0,1)$.
\subsubsection{Summary}
In the above, we have discussed the enhancibility of a Bell-diagonal state for RSP in two different cases. As we can see, for both cases, we come to the same conclusion that if the condition Eq.~(\ref{eqtest}) holds, there exists a optimal symmetric local amplitude damping channel with parametric
\begin{equation}
p_{\text{opt}}=\frac{c+\sqrt{c^2+4\left(c-c_3\right)}}{2+c+\sqrt{c^2+4\left(c-c_3\right)}},
\label{eqpopt}
\end{equation} use which we can enhance the corresponding Bell-diagonal state.

In Fig.~\ref{fig2}(a), we have plotted the region that represents the Bell-diagoal states satisfying the condition of Eq.~(\ref{eqtest}). In Fig.~\ref{fig2}(b), we depict the increment of RSP-fidelity of a one-dimensional class of Bell-diagonal states before and after the action of the optimal symmetric local amplitude damping channel. As we can see, even the classical state with coordinates $\left(-1, 0, 0\right)$ can be used for RSP (with non-zero RSP-fidelity) after a proper local operation.

\section{Conclusions}\label{seciv}
In this paper, we study the influence of local operations on the RSP-fidelity. We find that RSP-fidelity does not increase for any unital local operation, while the resource sates may be enhanced by suitable nonunital local operations. We use the symmetric local amplitude damping channel acting on Bell-diagonal state to exemplify that nonunital channel under specific condition can enhance the fidelity of RSP. We also use the same example to show that RSP-fidelity can vanish at instant under local Markovian decoherence, which is different from GMQD. At last, we give the optimal parameter of symmetric local amplitude damping channel for enhancing Bell-diagonal resource states [Eq.~(\ref{eqpopt})], and the criterion for the enhancibility of Bell-diagonal states [Eq.~(\ref{eqtest})].

We hope that our research will help to promote the future investigations on the power of local operation for enhancing quantum information protocols. Also we look forward to seeing more works of revealing the essence of increasing quantum correlation by local operation.

\section{Acknowledgments}
This research is supported by the NNSF of China, Grants No. 11075138.

\end{document}